\documentclass[10pt]{article}
\usepackage{amsmath}
\usepackage{amssymb}
\usepackage{amsthm}
\usepackage[ruled,vlined]{algorithm2e}
\usepackage{graphics}
\usepackage{graphicx}
\usepackage{epstopdf}
\usepackage{subfigure}

\def\sumi{\sum_{i=1}^n} 
\def\sumj{\sum_{j=1}^J} 
\def\sumk{\sum_{k=1}^K}

\newtheorem{prop}{Proposition}

\newtheorem{rem}{Remark}

\usepackage{setspace}

\def\etal{{\em et al. }}

\usepackage[total={17cm, 23.5cm},
top=3.5 cm, left=2.5 cm, right = 2.5 cm]{geometry}
\usepackage{authblk}

\title{Enhanced Cluster Computing Performance \\through Proportional Fairness}
\author[1]{Thomas Bonald}
\author[2]{James Roberts\thanks{The authors are members of the LINCS, Paris, France. See www.lincs.fr.}}
\affil[1]{T\'el\'ecom ParisTech, Paris, France}
\affil[2]{IRT System-X, Paris-Saclay, France}

\date{}
\begin{document}
  \maketitle

\begin{abstract}

The performance of cluster computing depends on how concurrent jobs share multiple data center resource types like CPU, RAM and disk storage. Recent research has discussed efficiency and fairness requirements and identified a number of desirable scheduling objectives including so-called dominant resource fairness (DRF). We argue here that proportional fairness (PF), long recognized as a desirable objective in sharing network bandwidth between ongoing flows, is preferable to DRF. The superiority of PF is manifest under the realistic modelling assumption that the population of jobs in progress is a stochastic process. In random traffic the strategy-proof property of DRF proves unimportant while PF is shown by analysis and simulation to offer a significantly better efficiency-fairness tradeoff.

\end{abstract}

\section{Introduction}
We consider a cluster of servers, exclusively reserved for executing a certain class of jobs, as a set of distinct resource pools. The resources in question include CPU, RAM, disk space and I/O bandwidth, for instance. A job must execute a certain number of tasks, each task requiring a given quantity of each type of resource. An example of this type of computing environment and its workload are described by Reiss \etal \cite{Reiss2012b}.  
Each task has its own particular requirements profile specifying how much of each resource must be reserved: some tasks require more CPU than RAM, others require more RAM than I/O bandwidth, etc. The issue at hand is how a central scheduler should initiate and run tasks in parallel to fairly and efficiently share cluster resources between jobs. We argue in this paper that such resource sharing should realize the objective of Proportional Fairness (PF), that is, it should maximize the sum over jobs of logarithms of their number of tasks.

The multi-resource sharing problem was formulated in a recent paper by Ghodsi \etal \cite{Ghodsi2011} under the simplifying assumption that resource pools are homogeneous and infinitely divisible between jobs. All tasks of a given job are assumed to have identical profiles and, rather than allocating resources to discrete tasks, a central scheduler is supposed to assign shares of each resource in proportion to the task profile.  
Sharing objectives are expressed in terms of a number of required properties. In particular, it is required that allocations be strategy-proof  in the sense that the owner of a job cannot gain a bigger share by lying about its actual task profile. The authors of \cite{Ghodsi2011} show that this property is not satisfied by many otherwise intuitively appealing strategies. They proceed to define an original strategy-proof allocation called Dominant Resource Fairness (DRF) and explain how this can be realized.

Other authors have since come up with alternative strategies, based on different notions of fairness (e.g., \cite{Dolev2012, Zeldes2013}), or adapting DRF to account for more realistic cluster constraints (e.g., \cite{Parkes2012, Psomas2013, Wang2013}). Gutman and Nisan \cite{Gutman2012} situate proposed fairness objectives like DRF \cite{Ghodsi2011} and ``no justified complaints'' \cite{Dolev2012} in a common economics framework and provide efficient polynomial time algorithms.   

We believe the economics literature cited in \cite{Gutman2012} does not in fact constitute an appropriate background for modelling cluster resource sharing.  This is because it completely ignores the dynamics of job arrivals and completions that characterize cluster workload \cite{Reiss2012b}. These dynamics have an obvious impact on job completion times and interact closely with the applied allocation algorithm since a job finishes more or less quickly depending on the resources it is allocated. The notion of strategy-proofness  must be revisited to account for the impact of false user requirements on the process of jobs in progress and their expected performance. 

In advocating PF we draw on our understanding of bandwidth sharing performance in networks.   PF was defined by Kelly \etal \cite{Kelly1998} as the bandwidth allocation that maximizes the sum of logs of flow rates, arguably thus realizing greater social welfare than alternative allocations like max-min fairness.  We consider rather the impact of the sharing objective on flow completion times in dynamic traffic. In wired networks, it turns out that completion times are not highly dependent on the type of fairness imposed and DRF (equivalent here to max-min fairness) and PF have similar performance \cite{Bonald2006}. In wireless networks, on the other hand, the radio resource is measured in time slots per second rather than bit/s with the achievable bit rate per time slot depending significantly on user position. Sharing a channel equally in time, as realized in HDR/HSDPA systems for instance, actually corresponds to PF sharing in terms of bandwidth. In a mixed wired/wireless network, PF realizes a significantly more favourable efficiency--fairness tradeoff than max-min fairness \cite{Bonald2011}.

To compare PF and DRF in dynamic traffic we adopt a simple Markovian traffic model assuming Poisson job arrivals and exponential job sizes. This simplifies analysis and clarifies the respective tradeoffs realized by the two allocation approaches. Simulation and analytical results confirm that PF performs significantly better than DRF. Given known insensitivity properties of fair resource sharing, we are confident that PF would be equally preferable under a more realistic traffic model  \cite{Bonald2006}.

In addition to discussing ideal sharing of infinitely divisible resources, Ghodsi \etal show how to approximately realize DRF in practice accounting for discrete, finite-size tasks. The scheduler algorithm preferentially launches tasks of the ``most deprived'' job. This is the job whose current share of its dominant resource is smallest. PF can be implemented similarly on re-defining ``most deprived'' in terms of the ideal shares determined for that allocation objective. The complexities of PF and DRF algorithms are similar and hardly constitute an implementation issue since the number of resource types to be shared is typically very small.

In the next section we define DRF and PF and demonstrate their respective sharing properties with respect to a static job population. We consider dynamic sharing in Section \ref{sec:dynamicshare} assuming a fluid model where resources are infinitely divisible. The task-by-task implementations of DRF and PF are compared in Section \ref{sec:taskassign}. Finally, related work is discussed in Section  \ref{sec:related} before we conclude.

\section{Multi-resource sharing}
\label{sec:staticshares}
We consider how multiple resources should be shared assuming a static population of jobs, each with a particular profile of per-task requirements. We adopt a fluid model where pools of resources are assume infinitely divisible and compare two allocation strategies: dominant resource fairness (DRF) and proportional fairness (PF).

\subsection{The fluid model}
We consider $J$ infinitely divisible pools of resources of respective capacities $C_j$, for $j=1,\ldots,J$, to be shared by  $n$ jobs indexed by $i$. Each task of job $i$ requires $A_{ij}$ units of resource $j$. Denoting by $\varphi_i$ the number of ongoing tasks of job $i$, we have the capacity constraints:
\begin{equation*}
\sum_{i=1}^n \varphi_i A_{ij} \le C_j,
\end{equation*}
for $j=1,\ldots,J$. In the fluid model, we assume tasks are infinitesimally small and jobs can run a sufficiently large number of them to attain a given resource allocation. It then makes more sense to normalize resource capacities to 1 with $a_{ij}=A_{ij}/C_j$ representing fractional requirements. The (now) real numbers $\varphi_1,\ldots, \varphi_n$ are then considered as \emph{task volumes} satisfying capacity constraints:
\begin{equation}\label{eq:capa}
\sum_{i=1}^n \varphi_i a_{ij} \le 1,
\end{equation}
for $j=1,\ldots,J$. The product $\varphi_i a_{ij}$ is the fraction of resource $j$ allocated to $i$. We also write capacity constraints \eqref{eq:capa} in matrix form: $\varphi a \le 1$, where inequality is understood component-wise.

We say that resource $j$ is \emph{saturated} if the corresponding capacity constraint is
attained, i.e., if $\sum_{i} \varphi_i a_{ij} = 1$. Let $a_i$ denote the vector whose $j^{th}$ component  is $a_{ij}$. We say job $i$ \emph{needs} resource $j$ if $a_{ij} > 0$ and we assume each job needs at least one resource. The \emph{dominant resource} of
job $i$ is that for which the normalized requirement $a_{ij}$ is largest.

\subsection{Desirable properties}
\label{sec:properties}

The following are desirable properties of a multi-resource allocation algorithm, as discussed by  Ghodsi \etal  \cite{Ghodsi2011}: 

\noindent \textbf{Pareto-efficiency. }An algorithm is Pareto-efficient if each job needs some resource which is saturated.

\noindent \textbf{Sharing-incentive.} An algorithm satisfies sharing-incentive if each job gets no less than a $1/n$ share of its dominant resource.

\noindent \textbf{Strategy-proof.} An algorithm is strategy-proof if the resource shares of any job do not increase if this job modifies its  
vector of resource requirements.

It proves useful to introduce two additional properties:

\noindent \textbf{Scale-invariance.} An algorithm is scale-invariant if resource shares remain the same when the vector of resource requirements of any job is multiplied by a scaling factor.

\noindent \textbf{Local-fairness.} An algorithm is said to be locally fair if resource shares $\varphi_i a_{i1}$ are all equal to $1/n$ when $J=1$.

Pareto-efficiency means no resources are needlessly left idle. Sharing-incentive means
each job gets no less than it would under a strictly fair allocation where a $1/n$ share of
each resource is allocated to every job. Strategy-proof means no job can
get a better allocation by lying about its requirements. Scale-invariance means
resource shares do not depend on the size of tasks (e.g., a scaling
factor of 2 is equivalent to merging two tasks into one) and may be considered
as a weaker form of strategy-proofness.  Local-fairness means the allocation 
reverts to simple fair sharing when there is only one resource.

\begin{rem} \label{rem:localfair} It is straightforward to verify that local-fairness is realized by any Pareto-efficient algorithm that also satisfies any one of sharing-incentive, strategy-proof or scale-invariance. This means that if it is not locally fair, a Pareto efficient  algorithm has none of the other properties. 
\end{rem}

\subsection{Dominant resource fairness}

We first note that simple max-min fairness is unsatisfactory. The max-min allocation is determined by the standard water-filling procedure: the $\varphi_i$ are increased from zero together until some resource is saturated; the allocations of any job that needs this resource are frozen while the task volumes of the remainder increase together until another resource is saturated; the procedure continues until all job allocations are frozen.

The resulting allocation is uniquely characterized by the fact that, for each job $i$, there
is some saturated resource $j$ needed by job $i$ such that $\varphi_i$
is maximum over all jobs that need resource $j$. In particular, for $J=1$, we have equal task volumes, $\varphi_i = 1/(\sum_{k=1}^n a_{k1})$,
for $i=1, \ldots, n$. Resource shares $\varphi_i a_{i1}$ are thus proportional to requirements and max-min fairness is \emph{not} locally fair. Since the algorithm is Pareto-efficient, Remark \ref{rem:localfair} shows it has neither sharing-incentive, strategy-proof nor scale-invariant properties.
DRF is weighted max-min fairness with weights chosen to correct these failings.
Weights are inversely proportional to the dominant resource requirements:
$$ w_i = \frac{1}{\max_{j=1,\ldots,J} a_{ij}},$$
for $i=1,\ldots,n$. 

The allocation is derived from a modified water-filling procedure where filling rates are proportional to the per-job
weights. It is the unique Pareto-efficient allocation such that, for each
job $i$, there is some saturated resource $j$ needed by job $i$ such that $\varphi_i/w_i$
is maximum over all jobs that need resource $j$. For $J=1$, we get equal
resource shares, $\varphi_i  a_{i1} = 1/n,$ for $i=1, \ldots, n$, so that DRF is locally fair. It is proved in \cite{Ghodsi2011} that DRF also satisfies the sharing-incentive and strategy-proof  properties. For completeness, we provide short proofs for these and the scale-invariant properties.
 
\begin{prop}
DRF satisfies the sharing-incentive property.
\end{prop}
\begin{proof}
For each job $i$, there is some saturated resource $j$ needed by job $i$ such
that $\varphi_i/w_i$ is maximum over all jobs that need resource $j$. In particular,
$$
\varphi_i\max_{l=1,\ldots,J} a_{il}=\frac{\varphi_i}{w_i}\ge \frac{1}{n}\sum_{k:a_{kj}>0} \frac{\varphi_k }{ w_k} \ge \frac{1}{n}\sum_{k:a_{kj}>0}  \varphi_k a_{kj} =\frac{1}{n}.
$$
\end{proof}

\begin{prop}
DRF is strategy-proof.
\end{prop}
\begin{proof}
Assume that job $i$ claims resource requirements $a'_i$ instead of $a_i$. Denote
by $w'_i$ the corresponding weight and by $\varphi'_i$
the resulting allocation. We compare
the water-filling processes of both systems with the same filling rates for jobs
other than $i$. If job $i$ is limited in the modified system first, then $\varphi'_i /w'_i\le \varphi_i /w_i$ 
so that the dominant resource $j$ of job $i$ in the original system satisfies:
$$
\varphi'_i a'_{ij}\le \frac{ \varphi'_i }{ w'_{i}} \le  \frac{ \varphi_i }{ w_{i}}=\varphi_i a_{ij}.
$$
If, on the other hand, job $i$ in the original system is limited first on resource $j$, then all jobs except $i$ that need
this resource get a better allocation in the modified system. It follows that:
$$
\varphi'_i a'_{ij}\le 1-\sum_{k \ne i} \varphi'_k a_{kj}\le 1-  \sum_{k \ne i} \varphi_k a_{kj} = \varphi_i a_{ij}.
$$
\end{proof}

\begin{prop}
DRF is scale-invariant.
\end{prop}
\begin{proof}
Assume that each job $i$ claims resource requirements $a'_i=\alpha_i  a_i$
instead of $a_i$, for some $\alpha_i > 0$. Let $w'_i=w_i/\alpha_i$   and 
$\varphi'_i=\varphi_i/\alpha_i$. The allocation $\varphi'$
satisfies the new capacity constraints $\varphi' a' \le 1$ and gives the same resource
shares as the original allocation $\varphi$. It remains to show that $\varphi'$ is DRF. We know that, for each job $i$, there exists
some saturated resource $j$ such that $\varphi_i/w_i$ is maximum over all jobs that need
resource $j$; since   $\varphi'_i/w'_i=\varphi_i/w_i$, this property is also satisfied by the allocation
$\varphi'$ which is therefore DRF.
\end{proof}

\begin{rem} \label{rem:bias}
Note that DRF also has the following somewhat surprising characteristic: the allocation may depend on resource requirements for non-saturated resources. Consider $J=2$ resources, $n=2$ jobs and $a_1=( 1, 1/3) $ and $a_2=( 1/2, 1) $. Then $\varphi=(2/3, 2/3)$ and resource 2 is not saturated. 
In the absence of resource 2 constraints, that is for $a_1=( 1, 0) $ and $a_2=( 1/2, 0) $, the allocation changes to $\varphi'=(1,2)$ corresponding to equal shares of resource 1. The allocation depends on requirements for resource 2 even though this is not a bottleneck. 
\end{rem}

\subsection{Proportional fairness}
\label{sec:pfstatic}
An allocation that is well-known to have excellent properties in the context of sharing network bandwidth is proportional fairness (PF) \cite{Kelly1998}.  The PF allocation  is defined by:
\begin{equation}\label{eq:pf}
\arg \max_\varphi \sumi \log\varphi_i,
\end{equation}
under capacity constraints (\ref{eq:capa}). By construction, it is Pareto-efficient.

Let $\nu$ be the vector of Lagrange multipliers associated with capacity constraints \eqref{eq:capa}. The Karuch-Kuhn-Tucker theorem implies that PF is the unique allocation $\varphi$ such that:
\begin{equation}\label{eq:kkt}
\frac{1}{ \varphi_i}=\sumj a_{ij}\nu_j,
\end{equation}
for $i=1,\ldots,n$, where 
\begin{equation}\label{eq:kkt2}
\nu_j\ge0 \quad \text{and}\quad \nu_j \left(\sumi \varphi_i a_{ij}-1\right)=0,
\end{equation}
for $j=1,\ldots,J$.  

From \eqref{eq:kkt} and \eqref{eq:kkt2} we deduce,
\begin{equation}\label{eq:kkt3}
\sumj \nu_j= \sumj \nu_j \left(\sumi \varphi_i  a_{ij} \right)=\sumi \varphi_i  \left( \sumj  a_{ij} \nu_j \right)= n.
\end{equation}
It follows that PF is locally fair since, for $J=1$, $\varphi_i a_{i1}=1/n$.

\begin{prop}
PF satisfies the sharing-incentive property.
\end{prop}
\begin{proof}
It follows from  \eqref{eq:kkt} and \eqref{eq:kkt3} that:
$$
\varphi_i \max_j a_{ij}=\frac {\max_j a_{ij}}{\sumj a_{ij}\nu_j}\ge \frac {1}{\sumj \nu_j}=\frac{1}{ n}.
$$
\end{proof}

\begin{prop}\label{prop:pfproof}
PF is \emph{not} strategy-proof.
\end{prop}
\begin{proof}
It is sufficient to show this with an example. Consider $n=2$ jobs and $J=2$ resources with $a_1=( 1/2, 1) $ and $a_2=( 1, 1/2) $. We then have $\nu=(1, 1)$ and $\varphi=(2/3, 2/3)$. If job 1 claims $2/3$  units of resource 1 per task instead of $1/2$, we have 
$a'_1=( 2/3, 1) $ and $a'_2=( 1, 1/2) $  yielding $\nu'=(2, 0)$ and $\varphi'=(3/4, 1/2)$. Job 1 receives a bigger share at the expense of job 2.
\end{proof}

\begin{rem} 
Unlike max-min fairness, it is not sufficient for a job to increase its resource requirements to increase its resource shares. 
Taking  $a'_1=( 1, 1) $ and $a'_2=( 1,1/2) $ in the above example yields $\nu'=(2, 0)$ and $\varphi'=(1/2, 1/2)$.  Both job 1 and job 2 receive smaller resource shares. 
\end{rem}

\begin{rem} A job can be sure to increase its resource shares only if it knows the resource requirements of the other jobs. Taking  $a_1=( 1/2, 1) $ and $a_2=( 1,1/3) $   yields 
$\nu=(1, 1)$ and $\varphi=(4/5, 3/5)$. If job 1 applies the same strategy as in the  example of Proposition \ref{prop:pfproof}, we have $a'_1=( 2/3,1) $ and $a'_2=( 1,1/3) $ yielding $\nu'=(2,0)$ and $\varphi'=(3/4, 1/2)$. Both jobs again receive fewer resources. 
\end{rem}

The above remarks suggest a user wishing to exploit PF would be hard-pressed to determine an optimal strategy. The impact of false resource claims can only be appraised by considering likely combinations of jobs in progress, as determined by a model of dynamic traffic.  The following proposition shows that PF does have the scale-invariance guarantee.

\begin{prop}
PF is scale-invariant.
\end{prop}
\begin{proof}
Assume that each job $i$ claims resource requirements $a'_{i}=\alpha_i a_i$  instead of $a_{i}$, for some $\alpha_i>0$. Let  $\varphi'_i=\varphi_i/\alpha_i$. 
The allocation $\varphi'$   satisfies the new capacity constraints $\varphi' a' \le 1$ and gives the same resource shares as the original allocation $\varphi$. 
Moreover, from \eqref{eq:kkt} and \eqref{eq:kkt2}, it can easily be verified that $\varphi'$ is the PF allocation under capacity constraints $\varphi' a' \le 1$. 
\end{proof}

Unlike DRF, PF is not biased by non-bottlenecked resources (see Remark  \ref{rem:bias}) since these clearly have no impact on the allocation.

\subsection{Utility maximization}
PF sharing maximizes overall utility under the assumption that the utility for each job $i$ of allocation $\varphi_i$ is proportional to $\log \varphi_i$.  One might consider allocations defined in terms of more general, differentiable and strictly concave utility functions $U(\cdot)$, i.e., \begin{equation}\label{eq:ut}
\varphi = \arg \max_{\varphi} \sumi U(\varphi_i),
\end{equation}
 under capacity constraints \eqref{eq:capa}. 
 
 For instance, $U=(\cdot)^\alpha/(1-\alpha)$, for $\alpha > 0$, defines the family of $\alpha$-fair allocations \cite{Mo2000}. Proportional fairness corresponds to the limit $\alpha \to 1$ while max-min fairness is realized in the limit $\alpha\to +\infty$. 

Let $\nu$ be the vector of Lagrange multipliers associated with capacity constraints \eqref{eq:capa}. The Karuch-Kuhn-Tucker theorem states that a utility-based  allocation is the unique allocation $\varphi$ such that:
\begin{equation}\label{eq:kktu}
\forall i=1,\ldots,n,\quad  {U'(\varphi_i)}=\sumj a_{ij}\nu_j,
\end{equation}
under constraints \eqref{eq:kkt2}. Since utility maximization allocations are necessarily Pareto-efficient, the following property shows that no utility-based allocation except PF can be either sharing-incentive, strategy-proof or scale-invariant (cf. Remark \ref{rem:localfair}). 

\begin{prop}\label{prop:usharing}
PF is the only locally fair utility-based allocation.
\end{prop}
\begin{proof}
Take a single resource and two jobs with resource requirements $a_{11}=\frac 1 2$  and $a_{21}=\frac 1 {2x}$ for some $x>0$. In view of \eqref{eq:kktu}, we have $U'(\varphi_1)={x}U'(\varphi_2)$.
If the allocation is locally fair, we have equal resource shares $a_{11}\varphi_1=a_{21}\varphi_2=\frac{1}{ 2}$ so that:
$$
\forall x>0,\quad U'(x)=\frac{U'(1)}{ x}.
$$
This implies that $U=\alpha \log(\cdot)+\beta$ for some constants $\alpha>0$ and $\beta$, which is PF.
\end{proof}

\section{Behaviour in dynamic traffic}
\label{sec:dynamicshare}
\label{sec:fluidmodel}

Jobs handled by a data center cluster arrive over time and remain active for a finite duration that depends on the resource shares they are allocated. Algorithm performance, perceived in terms of job completion times, can only be realistically appraised under such dynamic traffic. We adopt a simple Markovian traffic model and retain the assumption of infinitely divisible resources to illustrate the relative performance of DRF and PF.  
\subsection{Traffic model}
\label{sec:trafficmodel}

We consider the following dynamic traffic model. There are $K$ distinct classes of jobs. Jobs of class $k$ arrive at rate $\lambda_k$ and require the completion of $\sigma_k$ tasks on average. 
We denote by $n_k$ the number of class-$k$ jobs in progress and by $n=n_1+\ldots+n_K$ the total number of jobs. 
Each task of a class-$k$  job  requires $a_{kj}$  units of resource $j$. The capacity constraints are:
\begin{equation}\label{eq:capa2}
 \sumk n_k\varphi_k a_{kj}\le 1,
\end{equation}
for $j=1,\ldots,J$.
Let $\tau_k$ be the mean service time of one task of a class-$k$ job. Then 
$\mu_k=1/(\sigma_k\tau_k)$ is the completion rate of a  class-$k$ job with 
exactly one ongoing task and  $\varphi_k\mu_k$ is the  completion rate of a class-$k$ job. 

Under the fluid sharing model with exponential assumptions, i.e., Poisson job arrivals and exponentially distributed job sizes, the vector $\vec n=(n_1,\ldots,n_K)$ is a Markov process with component-$k$ birth rate $\lambda_k$ and death rate $n_k\varphi_k\mu_k$, for $k=1,\ldots,K$. 

Markovian assumptions are adopted for the sake of simplicity. We note however that insensitivity results for bandwidth sharing in networks suggest the performance of this simple model is indicative of much more general and realistic traffic models (e.g., see \cite{Bonald2006}). The impact of non-infinitely divisible resources is evaluated in the next section.

We assume allocations are instantaneously updated every time a new job starts or a job in progress completes. 
For DRF the allocation is determined by applying the water-filling algorithm. 
For PF we can determine the vector of Lagrange multipliers $\nu$ using the gradient descent method: 
\begin{equation}
\nu_j\gets \left(\nu_j+\gamma\left(\sumi  \psi_i{a_{ij}}-1\right)\right)^+, 
\label{eq:lagrange}
\end{equation}
where $\gamma>0$ is a sufficiently small step size and, for $i=1,\ldots,n$, $\psi_i= 1/ \sumj a_{ij}\nu_j$.

\subsection{Traffic capacity}

Tasks of  class-$k$ jobs arrive at rate $\lambda_k\sigma_k$ and last $\tau_k$ seconds on average. The corresponding traffic demand is thus $\rho_k=\lambda_k/\mu_k$. In dynamic traffic, the considered system may be unstable if demand is too high in the sense that the number of jobs in progress may then grow indefinitely. The set of loads for which the system remains stable defines a capacity region. The capacity region of DRF and PF (and alternative utility-based fairness allocations) can be directly deduced from known results for bandwidth sharing in networks \cite{Ye2003}. It is defined by the following inequalities:
\begin{equation}
\sumk \rho_k a_{kj}<1,
\label{eq:capacity}
\end{equation}
for $j=1,\ldots,J$. In other words, the load on each resource must be less than its normalized capacity.

\subsection{Performance}
\label{sec:fluidperf}

Users experience system performance through the time it takes to complete their job. As this is job size dependent, we prefer to present class-$k$ performance in terms of the mean service rate $\gamma_k$ defined as the ratio of the mean job size to the mean completion time.  By Little's law, we have the following expression:
$$
\gamma_k=\frac{\lambda_k}{ {\rm E}(n_k)}.
$$

We have evaluated the mean service rates $\gamma_k$ for DRF and PF allocations by simulating some simple system configurations. 
We suppose each job requires quantities of two resources, CPU and RAM say, with relative per-task requirements $a_1=( 1, 1/10) $ and $a_2=( 1/10, 1) $.  We set $\mu_1=\mu_2=1$ and consider two demand patterns: \emph{balanced} with  $\lambda_1=\lambda_2 = \lambda/2$, where $\lambda$ is the total demand, and \emph{unbalanced} with $\lambda_1=\frac 3 4\lambda$ and $\lambda_2=\frac 1 4\lambda$. For balanced demand the load of both resources is $\frac \lambda 2(1+\frac 1{10})$ whereas the unbalanced case yields a CPU load of $\frac \lambda 4( 3  + \frac 1 {10})$ and a RAM load of $\frac \lambda 4(  \frac 3 {10}+1 )$, nearly 3 times smaller.

\begin{figure}[h]
\begin{center}
\subfigure[balanced load ($\rho_1= \rho_2$)]{\includegraphics[width=7cm]{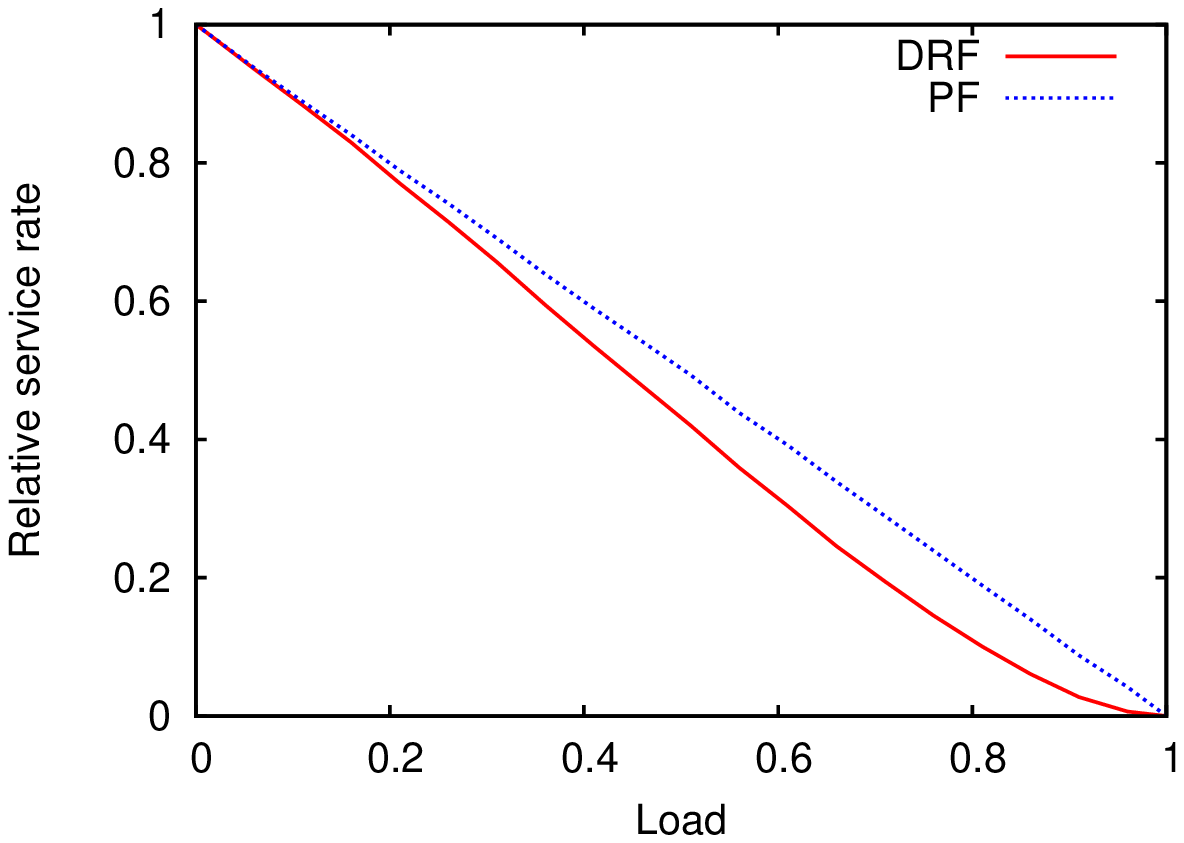}\label{fig:exp-a}}
\subfigure[unbalanced load ($\rho_1=3 \rho_2$)]{\includegraphics[width=7cm]{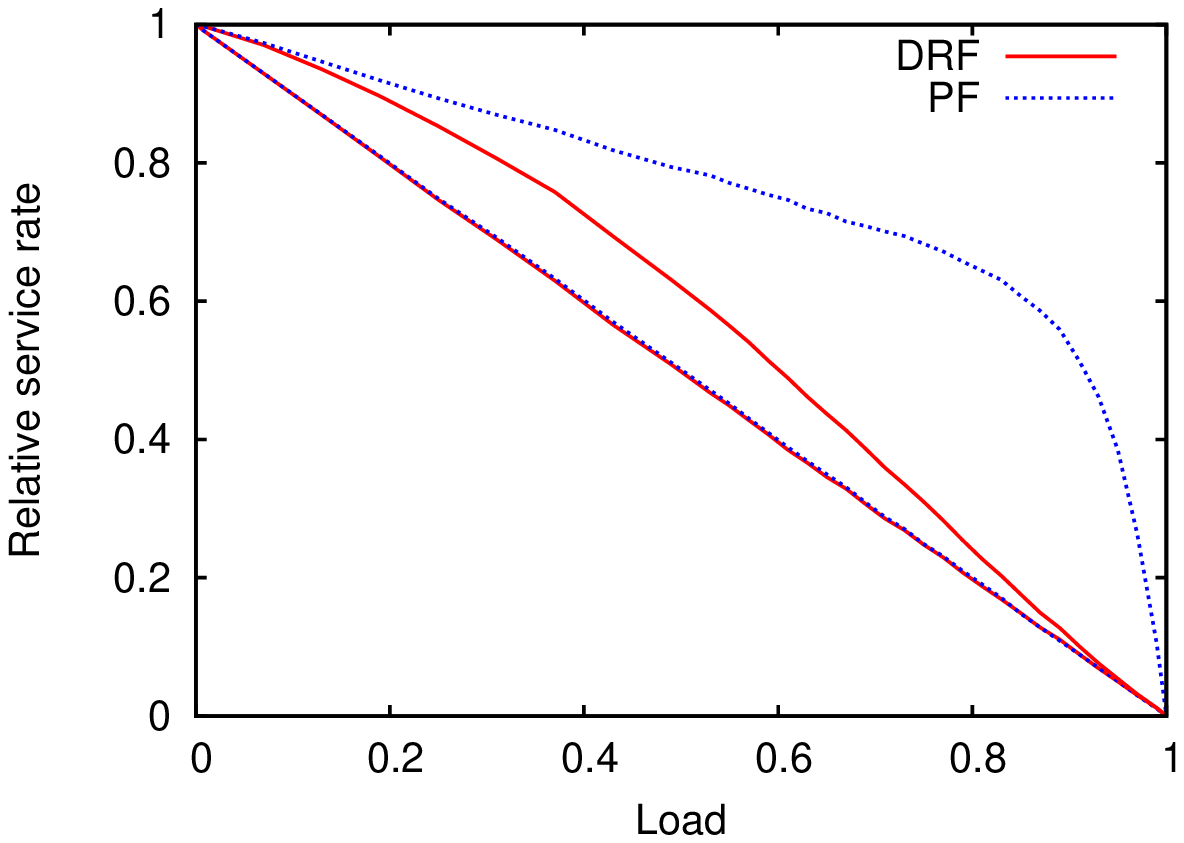}}
\caption{\label{fig:exp}Service rate against load of resource 1 (CPU).}
\end{center}
\end{figure}

Figure \ref{fig:exp} compares the mean relative service rate provided by DRF and PF allocations as a function of CPU load. 
Note that the service rate is 1 (in normalized units) at load zero, when never more than 1 job is in progress, and tends to zero as load attains system capacity. For a single resource, the service rate would in fact be $1 - \varrho$ for both DRF and PF, where $\varrho$ is the resource load. This is so because allocations are locally fair and therefore realize a processor sharing system. 

Both graphs in Figure \ref{fig:exp} present 4 curves representing $\gamma_k$ for $k=1,2$ for DRF and PF.  
In the left hand plot the curves for both classes for each fairness objective coincide by symmetry while PF outperforms DRF, especially at high load. The right hand plot for unbalanced demand shows that $\gamma_2^{(PF)}$ is significantly higher than $\gamma_2^{(DRF)}$ while the service rates $\gamma_1$ for the CPU constrained class are practically the same. This occurs basically because of the bias exhibited by DRF but not PF when resource 2 (RAM) is not saturated, as noted in Remark \ref{rem:bias}.
 
\begin{figure}[h]
\begin{center}
\subfigure[balanced load ($\rho_1=\rho_2$)]{\includegraphics[width=7cm]{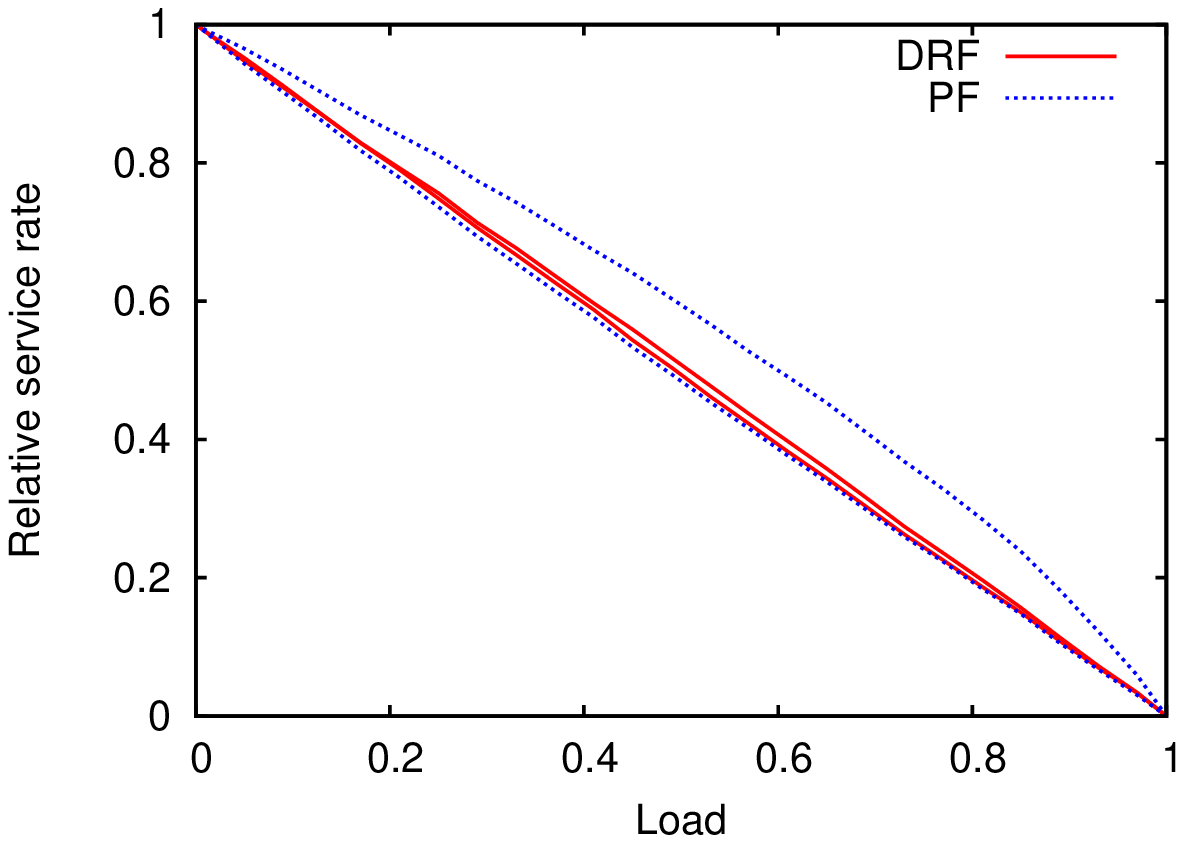}}
\subfigure[unbalanced load ($\rho_1=3\rho_2$)]{\includegraphics[width=7cm]{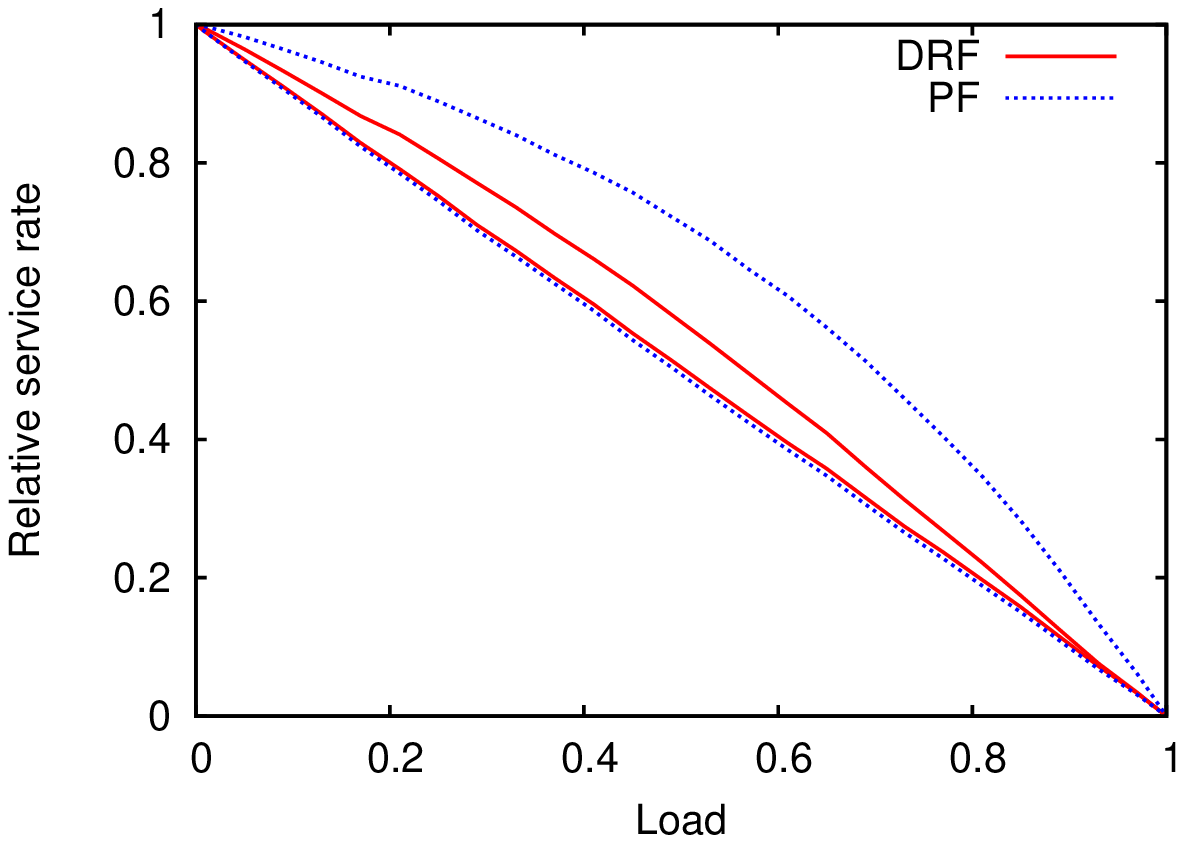}}
\caption{\label{fig:exp2}Service rate against load of resource 1 (CPU).}
\end{center}
\end{figure}

Figure \ref{fig:exp2} presents a case where PF also performs better than DRF at low load. To be precise, the service rate of class 2 is significantly greater at the cost of a slight decrease of that of class 1. The class profiles for this case are  $a_1=(1,1/10)$ and $a_2=(1/2, 1)$.

The above simple examples clearly do not prove that PF is always better than DRF. To gain insight first consider the light traffic behaviour of a general 2-class, 2-resource type system. Suppose, without loss of generality in view of scale invariance, job classes have task profiles  $a_1=(1,\alpha)$ and $a_2=(\beta,1)$ with $\alpha \le \beta \le 1$. 
Considering the evolution of the job population as $\rho_1 \to 0$ and $\rho_2 \to 0$, we derive the following second order approximations for DRF and PF: 
$$\gamma_1^{(\textsc{DRF})} \approx  1 - \rho_1 - \beta \rho_2 \quad  \mathrm{ and }  \quad \gamma_2^{(\textsc{DRF})} \approx  1 -\beta \rho_1 - \rho_2.$$  
$$\gamma_1^{(\textsc{PF})} \approx  1 - \rho_1 -  \beta \rho_2 \quad \mathrm{ and } \quad  \gamma_2^{(\textsc{PF})} \approx  1 - \max \left(2-1/\beta, \alpha\right) \rho_1 - \rho_2.$$
We see that $\gamma_1$ has the same light traffic behaviour for both allocations while $\gamma_2$ decreases more rapidly for DRF than for PF since $\max \left(2-1/\beta, \alpha\right) \le \beta$.

Consider now the same profiles in heavy traffic with $\varrho = \rho_1 + \beta \rho_2  > \alpha \rho_1 + \rho_2$, so that resource 1 has higher load than resource 2, and $\varrho \to 1$. PF may then be shown to yield $\gamma_1^{(\textsc{PF})} \approx 1-\varrho$ and $\gamma_2^{(\textsc{PF})} \approx (1-\varrho)/\alpha$. The water filling definition of DRF, on the other hand, will equalize the service rates of both classes $\gamma_1^{(\textsc{DRF})} \approx \gamma_2^{(\textsc{DRF})} \approx 1-\varrho$. The latter approximation derives from the underlying ``processor sharing'' system emulated by DRF in heavy, unbalanced traffic. Note that the behaviour of DRF under balanced load with symmetric profiles is rather worse than this in heavy traffic (cf. Fig. \ref{fig:exp-a}). This is because the allocation is constrained by the resource with the heaviest instantaneous load, $\arg \max_j \left( n_1 a_{1j} + n_2 a_{2j}\right)$. 

\begin{figure}[h]
\begin{center}
\subfigure[balanced load ($\rho_1=\rho_2=\rho_3=\rho_4$)]{\includegraphics[width=7cm]{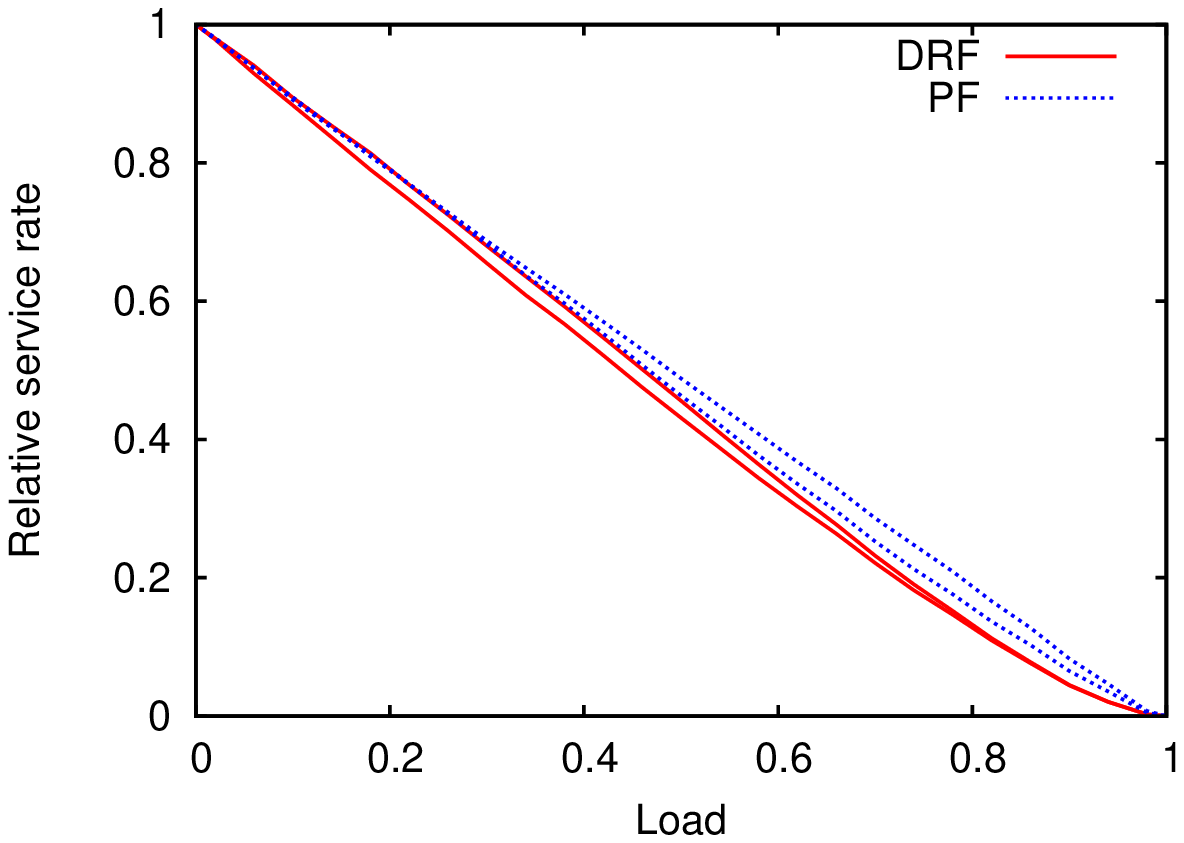}}
\subfigure[unbalanced load ($\rho_1=\rho_2=3\rho_3=3\rho_4$)]{\includegraphics[width=7cm]{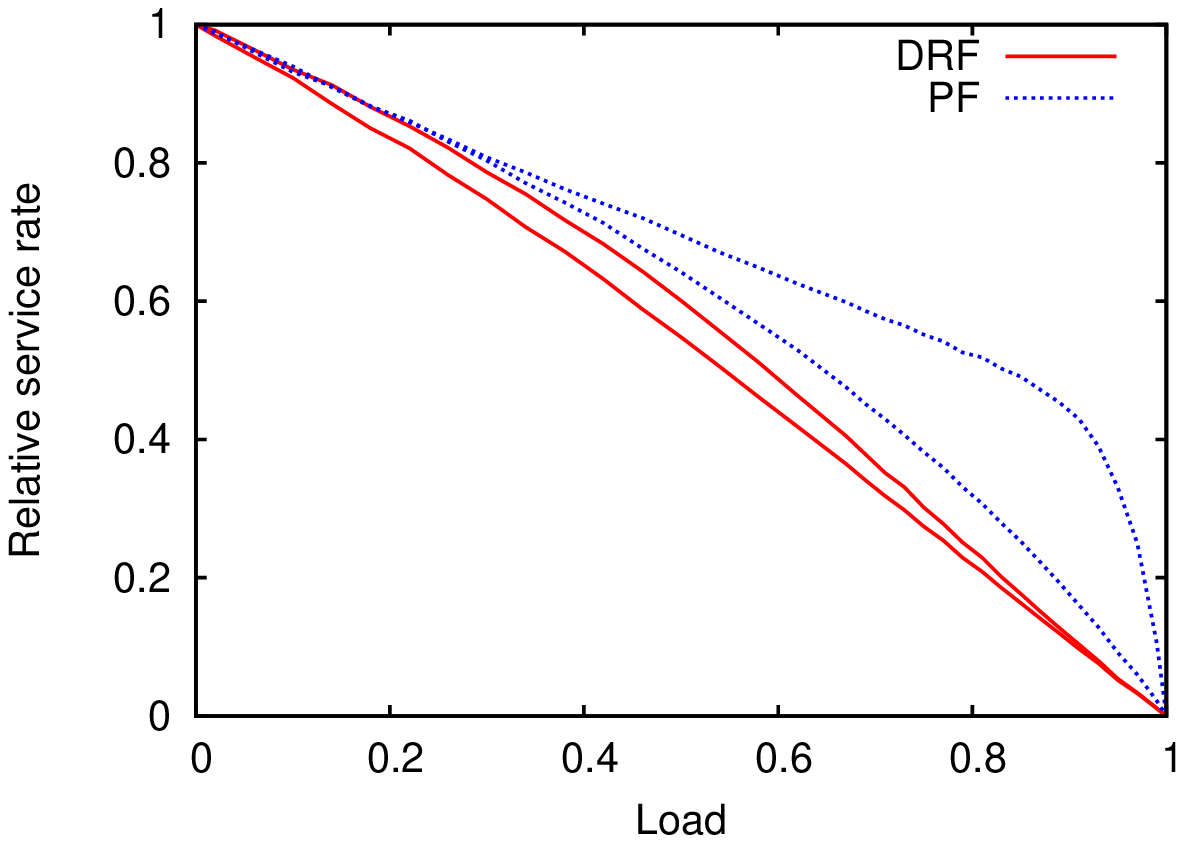}}
\caption{\label{fig:exp5}Service rate against load of resource 1 (CPU).}
\end{center}
\end{figure}

In Figure \ref{fig:exp5}, we plot results for 4 job classes with profiles $a_1= (1,1/10)$, $a_2= (1,1/2)$, $a_3= (1/2,1)$ and $a_4= (1/10,1)$. The figures plot the service rates of classes 3 and 4 whose dominant resource is resource 2 (RAM). These results confirm the trends shown in previous cases are maintained when the traffic mix is more varied.

 We have simulated several other configurations including cases with three types of resource. The results confirm the trends illustrated in the above plots: PF slightly outperforms DRF for balanced traffic; the gain is more significant for jobs whose dominant resource has relatively lower load; the gain comes at the cost of an imperceptible loss of service rate for jobs whose dominant resource is most heavily loaded. The latter is approximately $1-\varrho$ for both DRF and PF where $\varrho$ is the resource load, as predicted by the underlying processor sharing model, except for DRF under balanced load. DRF tends to behave poorly in heavy, balanced traffic since it penalizes efficiency in favour of strict fairness.
 
\subsection{Strategy-proofness}
The strategy-proof property in Section \ref{sec:properties} relates to a static population of jobs. We demonstrate in this subsection that the absence of this property for PF is \emph{not a disadvantage} in dynamic traffic. Even if some job knows some statistics about  the  resource requirements of competing jobs, it does not know how many jobs are active and thus can hardly define a consistent strategy to improve its service rate by lying about its resource requirements. 

Consider the following example. There are $J=2$ resources and $K=2$ classes with 
resource requirements $a_1=(1,\alpha)$ and $a_2=(\alpha,1)$  for some $\alpha\in (0,1)$. There is a single, permanent class-1 job. 
The only way for this job  to increase its service rate is to claim $\beta\in (\alpha,1)$ 
  units of resource 2 per task. 
Under PF,  its service rate  in the presence of $n$ other jobs is:
 $$\varphi_1(n)=\min\left(\frac{1}{ n} \frac{1-\alpha }{1-\alpha\beta},\frac 1 {\beta(n+1)}\right),$$
 which is maximum for:
 $$
 \beta=\frac 1{1+n(1-\alpha)},
 $$
 Note that this value of $\beta$ lies in $(\alpha,1)$ as long as $n\le 1/\alpha$.
 Each  class-2 job gets $\varphi_2(n)=\frac{1}{ n}(1-\beta\varphi_1(n))$.
The number of class-2 jobs forms a birth-death process with birth rate $\lambda_2$ and death rate $\mu_2 n\varphi_2(n)$ in state $n$. This process is stable 
whenever $\rho_2<1$ and then has a certain stationary distribution $\pi$ yielding the mean service rate:
$$
\gamma_1^{(PF)}=\sum_n \pi(n) \varphi_1(n).
$$
Figure \ref{fig:1flow} plots $\gamma_1^{(PF)}$ against $\rho_2$ for $\alpha=\frac 1 {10}$ and three values of $\beta$, chosen to maximize the service rate $\varphi_1(n)$ assuming $n=1,2$ and 5, respectively. We observe that no benefit is gained from any of these static optimizations demonstrating that PF is indeed strategy proof in this scenario.

\begin{figure}[h]
\begin{center}
\includegraphics[width=7.5cm]{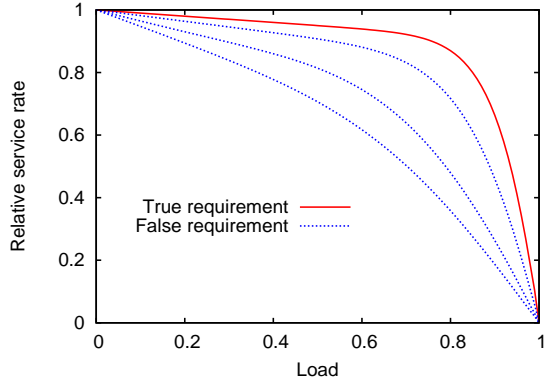}
\caption{\label{fig:1flow}Mean service rate of the test flow with true and false requirements that are optimal for 1, 2 and 5 other flows (from bottom to top).}
\end{center}
\end{figure}

If, in an alternative scenario, the requirements of all jobs in a given class $k$ were falsely declared with $a'_{kj}\ge a_{kj}$, for $j=1,\ldots,J$, and strict inequality for at least one resource, it is clear from (\ref{eq:capacity}) that the capacity region would be smaller for both DRF and PF. Simulation results shown in Figure \ref{fig:strategy2} confirm that service rates at stable loads are reduced for both DRF and PF. The figure plots the class 1 service rate for the configuration of Figure \ref{fig:1flow} with balanced load. The dashed lines result from the same values of $\beta$ used in Figure \ref{fig:1flow} confirming there is no winning strategy. Traffic capacity is limited by $(1+\alpha)/(1+\beta)<1$. Of course, class 2 jobs also suffer from the false claims of class 1 in all these cases.

\begin{figure}[h]
\begin{center}
\subfigure[DRF]{\includegraphics[width=7cm]{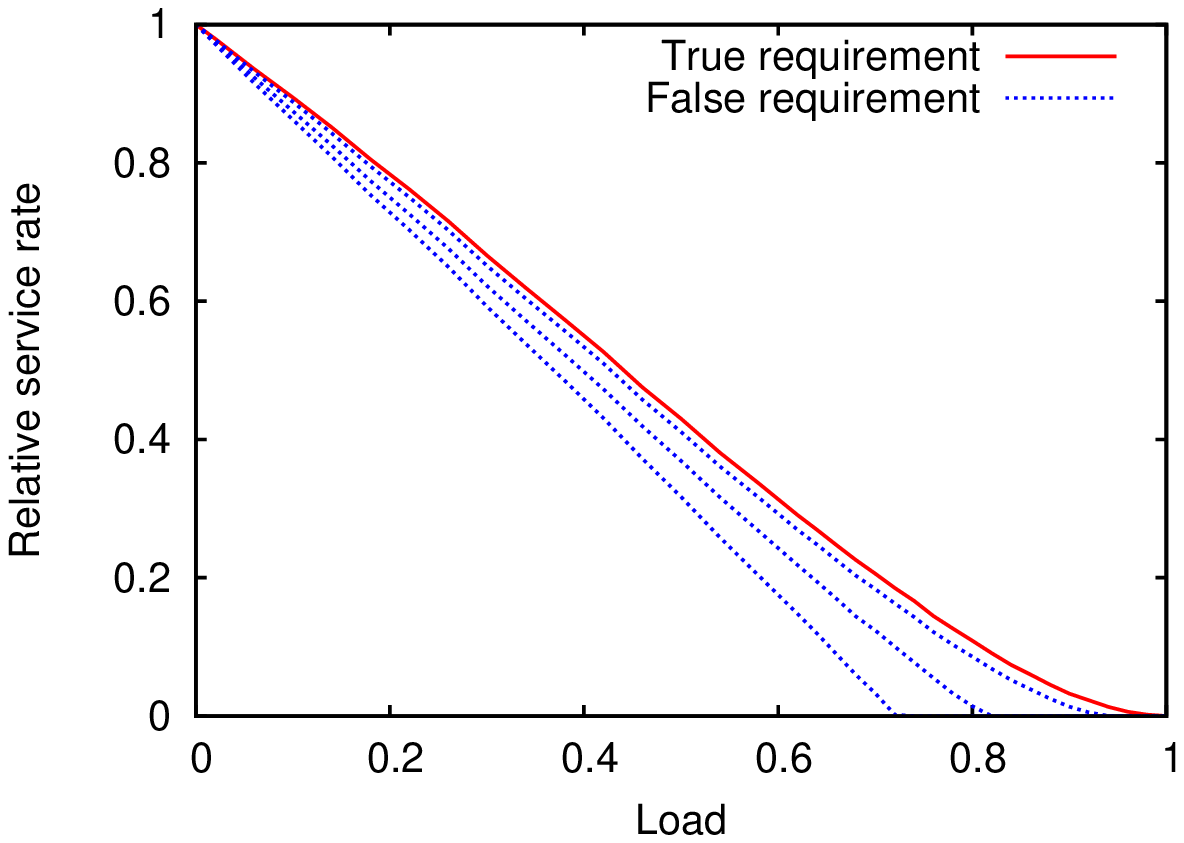}}
\subfigure[PF]{\includegraphics[width=7cm]{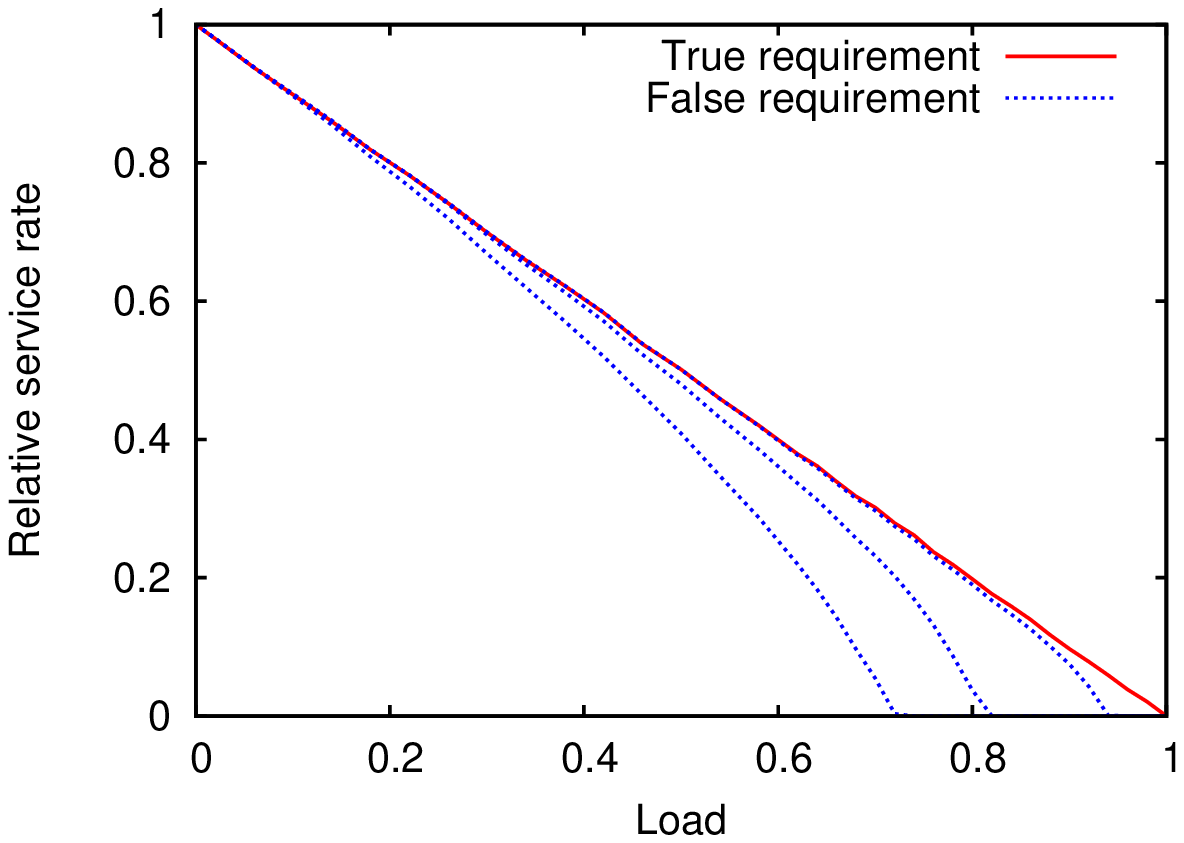}}
\caption{\label{fig:strategy2}Mean service rate of class 1 flows against CPU on for false class 1 requirements as in Fig. \ref{fig:1flow}}
\end{center}
\end{figure}

\section{Task assignment algorithms}
\label{sec:taskassign}
While the fluid models analyzed above are useful in appraising the relative merits of alternative fairness objectives, it is necessary in practice to account for the fact that tasks are not infinitely divisible. In this section we relax this assumption. We discuss the implementation of task-based versions of DRF and PF and compare their performance in simulation experiments.

\subsection{``Serve the most deprived job''}
We follow the proposal in \cite{Ghodsi2011} in seeking to realize the fairness objective by preferentially launching tasks of the most deprived job with respect to the objective fairness criterion. Let $\varphi_i$ now denote the integer number of tasks in progress for job $i$. For DRF,  the most deprived job is that for which $\varphi_i \max_j\{ a_{ij}\}$ is minimal. For PF, the job to serve is that for which $\varphi_i \sumj a_{ij}\nu_j $ is minimal, where the $\nu_j$ are Lagrange multipliers computed iteratively using (\ref{eq:lagrange}). 

Tasks are assigned when jobs arrive, if resources are available, or when other tasks end. If available resources are insufficient to accommodate a task of the most deprived job, allocations are frozen until a sufficient number of other tasks end or some other job takes over the most-deprived status. Note that we continue to assume resources are pooled though the assignment algorithm could be adapted to account for additional practical constraints such as assigning CPU and RAM on the same physical server, say,  \cite{Psomas2013}, or meeting other compatibility constraints \cite{Ghodsi2013}.

\subsection{Performance}
We have simulated task-based allocations in the following configuration. Two classes of job require different amounts of CPU and RAM in the following proportions: $a_1=(1,1/10)$ and $a_2=(1/10,1)$, as previously considered in Section \ref{sec:fluidperf}. Jobs of both classes require the execution of 500 tasks while CPU and RAM both have capacity 100. This means a single job of either class in an otherwise empty system will have 100 tasks in progress until 400 are completed and then 99, 98 and so on, until the last task finishes. This is a significant discrepancy with respect to the fluid model where the service rate would remain constant at 100 until job completion. With respect to previous notation, we have $\sigma_k=500$, $\tau_k=1/5$, $C=100$ and therefore $\mu_k=C/(\sigma_k \tau_k)=1$, for $k=1,2$.

\begin{figure}[h]
\begin{center}
\subfigure[Exponential task size]{\includegraphics[width=7cm]{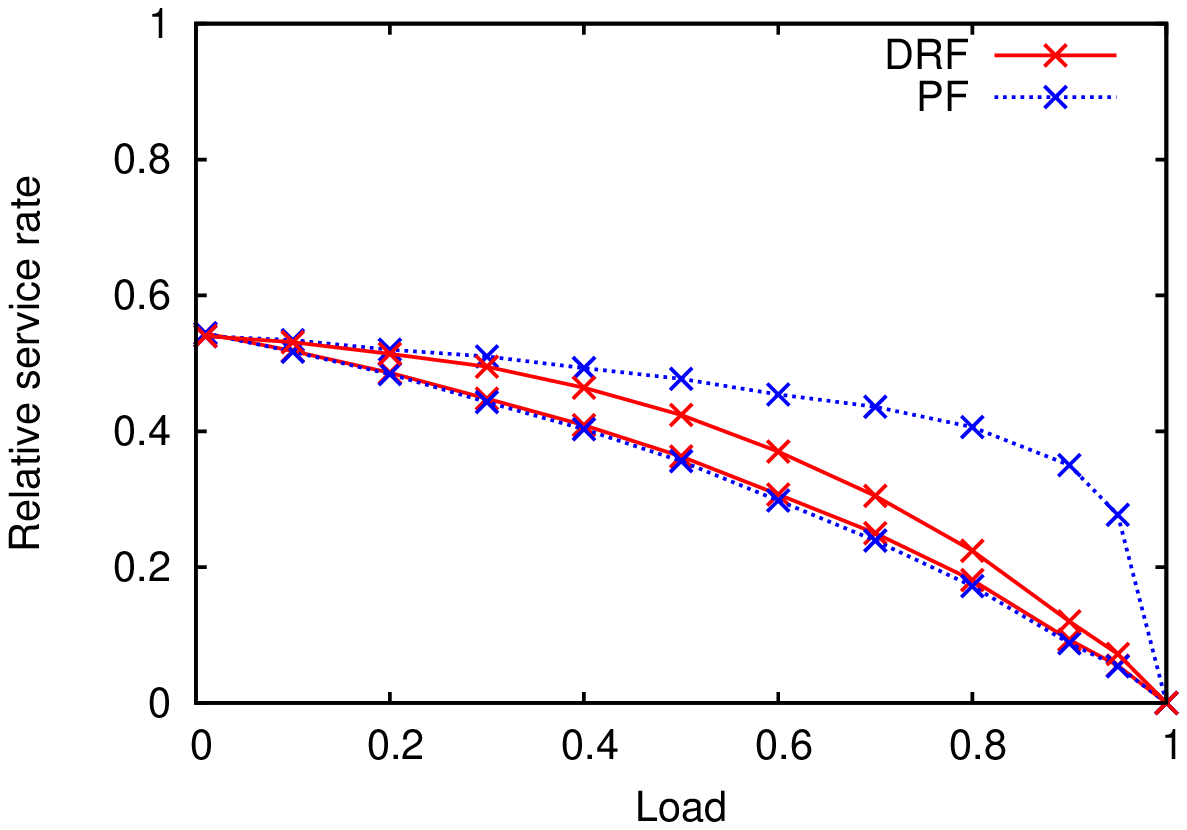}}
\subfigure[Erlang-20 task size]{\includegraphics[width=7cm]{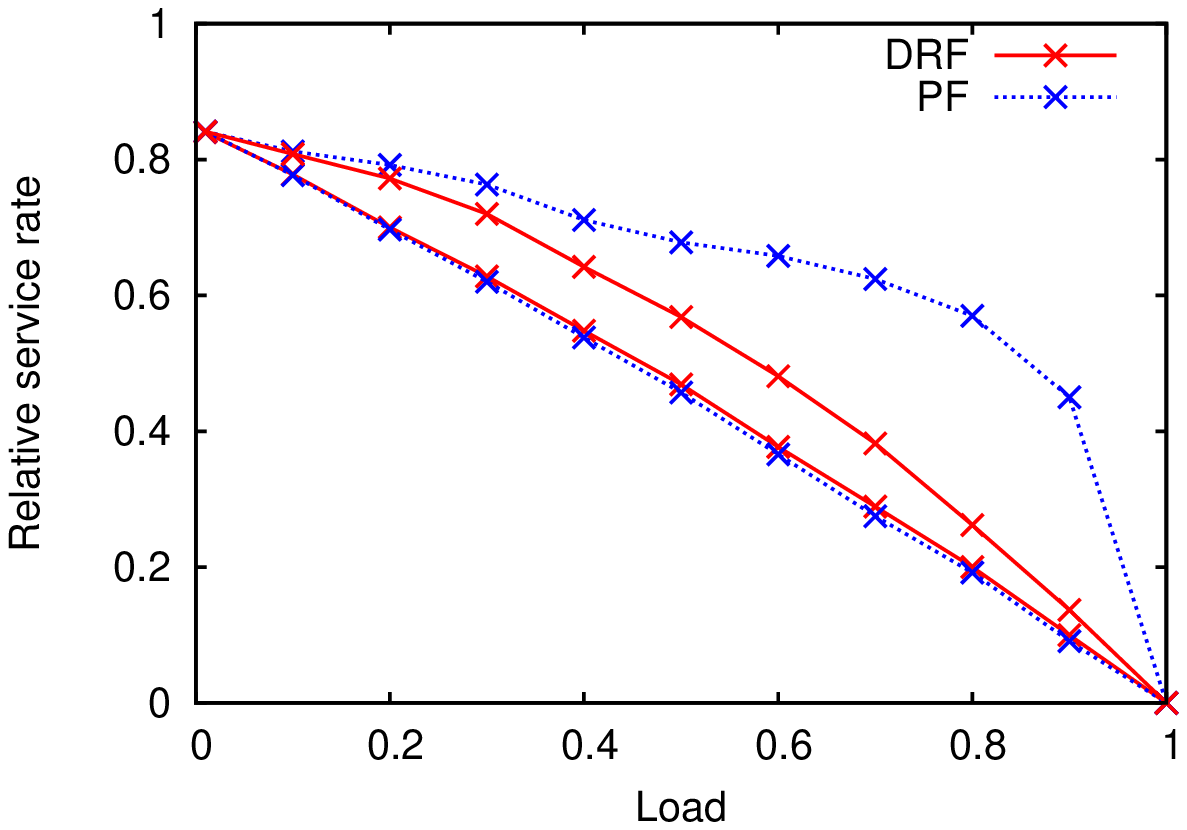}}
\caption{\label{fig:taskexp}Service rate against load of resource 1 (CPU) for unbalanced trafic ($\rho_1=3\rho_2$).}
\end{center}
\end{figure}

Figure \ref{fig:taskexp} shows the simulation results for DRF and PF allocations under unbalanced load, $\rho_1=3\rho_2$. The left hand plots relate to exponentially distributed task sizes. Compared to the right hand plots of Figure \ref{fig:exp} there is a significant reduction in service rate for low loads due to the discrepancy noted above. The service time of an isolated job in the fluid model is the sum of 500 task inter-departure intervals at rate 100 compared, in the task-based system, to 401 intervals at rate 100 plus one at rate 99, one at rate 98, \ldots, and 1 at rate 1. The ratio of service times is thus 
$$\frac{500}{100} /  \left( \frac{401}{100}+\frac{1}{99}+\frac{1}{98}+\ldots+1 \right) \approx 0.54,$$
corresponding to the zero load service rate seen in the figure. 

The discrepancy is smaller for task service times that are more deterministic. The right hand plot relates to an Erlang-20 distribution (i.e., the sum of 20 exponential random variables) with the same mean yielding service rates much closer to those of the fluid model. The fluid model is also a good fit when the number of tasks per job is much bigger than the resource capacity (results not shown).

Despite the discrepancy, it is clear that the relative advantage of PF over DRF, observed and analyzed in Section \ref{sec:fluidmodel},  is preserved under the task-based model. PF realizes a more favourable efficiency-fairness tradeoff, especially when demand for some resource is relatively low.

\section{Related work}
\label{sec:related}
A number of authors have extended the DRF concept introduced by Ghodsi \etal \cite{Ghodsi2011}.  Parkes \etal \cite{Parkes2012} generalize DRF notably to account for per-job sharing weights. Psomas and Schwartz \cite{Psomas2013}  take account of the fact that task resources must fit into a single machine, bringing bin packing considerations to the allocation problem. Wang \etal \cite{Wang2013} further generalize this approach by accounting for heterogeneous server capacity limits. Bhattacharya \etal \cite{Bhattacharya2013} add the notion of hierarchical scheduling to DRF: cluster resource allocations account also for shares attributed to the departments that own particular jobs. It is noteworthy that DRF is now implemented in the Hadoop Next Generation Fair Scheduler\footnote{http://hadoop.apache.org/docs/r2.3.0/hadoop-yarn/hadoop-yarn-site/FairScheduler.html}. Lastly, we note the proposal to use DRF in a different application setting, namely sharing middlebox resources in software routers \cite{Ghodsi2012, Wang2013b}.

An alternative multi-resource fairness criterion is proposed by Dolev \etal \cite{Dolev2012}. The authors introduce the notion of \emph{entitlement} corresponding to the share of each resource to which a job should minimally have access. The objective is then to define an allocation such that users do not receive less than their entitlement on at least one bottlenecked resource and therefore have ``no justified complaints''. The resulting allocation, named  bottleneck-based fairness (BBF), has been further developed by  Zeldes and Feitelson \cite{Zeldes2013} who propose an on-line algorithm. The multi-resource fairness problem is defined more formally by Gutman and Nisan who generalize the DRF and BBF definitions somewhat  \cite{Gutman2012}. These authors identify similarities with prior work on fairness in the economics literature and propose polynomial complexity algorithms drawn from that source. Joe-Wong \etal  \cite{Wong2012} propose alternative generalized definitions of what constitutes a fair multi-resource allocation identifying notions of ``Fairness on dominant shares'' (FDS) and ``Generalized fairness on jobs'' (GFJ).  

All the work cited in the above two paragraphs considers sharing between a fixed population of jobs and ignores the impact of traffic dynamics due to arrivals over time of finite size jobs. They do not therefore evaluate the performance of allocations in terms of realized job completion times which is arguably the most relevant criterion in appraising their effectiveness.

There is an abundant literature on sharing network bandwidth. PF and, more generally, the idea that bandwidth sharing between a fixed set of flows should maximize a sum of per-flow utilities were notions introduced by Kelly \etal in 1998 \cite{Kelly1998}. In that work, and subsequent contributions by many authors on network utility maximization, the population of flows in progress is assumed fixed. We believe this to be inappropriate, like the assumption that the population of jobs is fixed when evaluating cluster sharing.  It does not make much sense to measure the utility of a flow as a function of its rate (e.g., the log function for PF) on recognizing that the population of flows in progress, and consequently the rate of any ongoing flow, changes rapidly and significantly throughout the lifetime of that flow. 

The study of bandwidth sharing under a dynamic model of traffic also began at the end of the 90s \cite{Massoulie2000}. The most significant results derived over an ensuing period of several years are summarized in the paper by Bonald \etal \cite{Bonald2006}. It turns out that PF brings significantly better flow completion time performance than max-min fairness for a network with a mix of wired and wireless links (e.g., see Fig. 17.14 in \cite{Bonald2011}). Flows in such networks require unequal resource shares, like jobs in the compute cluster, so that these bandwidth sharing results anticipate the excellent performance of PF demonstrated here.

\section{Conclusions}
\label{sec:conclude}

Sharing multiple compute cluster resources like CPU and RAM is in many ways analogous to sharing bandwidth in a network where each link is considered as a distinct resource. The analogy is closest with respect to a wireless network where flows have different resource requirements on wired and radio links, respectively. We know proportional fairness (PF) performs well in such networks and have therefore sought to evaluate its suitability for sharing cluster resources. Our results confirm that PF is indeed preferable to so-called dominant resource fairness (DRF), the sharing criterion currently applied by the Next Generation Hadoop Fair Scheduler, for example.  

DRF was developed to realize a set of sharing properties derived essentially from micro-economics considerations. In particular, DRF is claimed to be strategy proof in the sense that a job cannot gain better service by announcing false requirements. We have argued here that these properties must be revisited to account for the dynamics over time of the job population. In this context, an algorithm is strategy proof if the expected completion time of a job cannot be reduced by false announcements, given that the user is necessarily unaware of the precise mix of competing jobs. 

In the dynamic setting, PF is strategy proof to the same degree as DRF: it is not possible to improve an allocation by declaring false requirements. PF realizes a better fairness--efficiency tradeoff than DRF over the range of possible loads. This is especially true when resource load is unbalanced with some resources remaining underloaded as others become saturated. DRF is then particularly deficient in overly limiting use of the underloaded resource in order to satisfy its particular notion of fairness. 

PF can be realized with the same, low complexity as DRF by adapting its ``serve-the-most-deprived-job'' policy. When resources are freed as tasks complete, they are assigned to new tasks of the job whose current allocation is furthest from the PF ideal. This ideal allocation is easily determined as the result of a simple constrained optimization problem.

It remains to adapt the PF allocation to account for practical placement constraints that occur in data centers. These include the requirement to allocate CPU and RAM on the same server, for example, or to use a specific type of server for some jobs. It is also of interest to incorporate hierarchical scheduling criteria whereby the departments originating certain classes of jobs have overall service rate guarantees. Available trace data suggests cluster resources are currently significantly over-provisioned so that optimized sharing is hardly an issue in practice. In future, however, one would expect tighter energy management leading to a requirement for more closely controlled performance. The key to this would be analytical modelling enabling one to predict performance given system capacity and load, by applying the robust engineering techniques developed over recent years for networks, for example.

\bibliographystyle{alpha}
\bibliography{multifair} 

\end{document}